\documentclass[10pt]{article}

\usepackage[utf8]{inputenc}
\usepackage{microtype}
\usepackage{mathtools,amssymb,mathrsfs,amsthm}   
\usepackage[table,svgnames]{xcolor}
\usepackage[hyperindex,breaklinks,colorlinks=true,linkcolor=black,citecolor=MidnightBlue,urlcolor=MidnightBlue]{hyperref}
\usepackage{xspace}
\usepackage{graphicx}
\usepackage{tikz}
\usetikzlibrary{chains,arrows,arrows.meta} 
\usepackage{fullpage}
\usepackage{floatpag}
\usepackage{marvosym}
\usepackage{framed}
\usepackage{breakurl}

\usepackage{algorithm}
\usepackage{algorithmicx}
\usepackage{algpseudocode}
\usepackage{comment}
\usepackage{verbatim}

\newtheorem{theorem}{Theorem}[section]
\newtheorem{lemma}[theorem]{Lemma}

\newtheorem{corollary}[theorem]{Corollary}
   
\newtheorem{claim}[theorem]{Claim}

\newtheorem{definition}{Definition}%
\newtheorem{remark}{Remark}

\usepackage{svn-multi}
\svnidlong
{$HeadURL: $}
{$LastChangedDate: $}
{$LastChangedRevision: $}
{$LastChangedBy: $}

\newcommand{\eps}{\varepsilon}
\newcommand{\poly}{\operatorname{poly}}

\newcommand{\neigh}{\mathcal{N}}

\newcommand{\Hamming}{\Delta}
\newcommand{\Weight}{\omega}

\newcommand{\col}{\mathsf{Collision}}
\newcommand{\silence}{\mathsf{Silence}}
\newcommand{\single}{\mathsf{SingleSender}}

\newenvironment{theorem-repeat}[1]{\begin{trivlist}
\item[\hspace{\labelsep}{\bf\noindent Theorem \ref{#1} }]\em }%
{\end{trivlist}}

\begin{document}

\title{\textbf{Noisy Beeping Networks}}

\author{%
Yagel Ashkenazi\thanks{\texttt{123yagel@gmail.com}. Supported in part by the Israel Science Foundation (ISF) through grant No.\@ 1078/17.} 
\and 
Ran Gelles\thanks{\texttt{ran.gelles@biu.ac.il}. \textbf{Corresponding Author}. Supported in part by the Israel Science Foundation (ISF) through grant No.\@ 1078/17  and in part by the U.S.-Israel Binational Science Foundation (BSF) through grant No.\@ 2020277.} 
\and 
Amir Leshem\thanks{ \texttt{amir.leshem@biu.ac.il}. Supported in part by the Israel Science Foundation (ISF) through grant No.\@ 1644/18.}
}

\date{Faculty of Engineering, Bar-Ilan University}

\maketitle

\begin{abstract}
Beeping networks were suggested by Afek et al.~[Science~331(6014), 2011] as a way to abstract the communication of exceedingly simple computational devices such as the brain cells of flies or ants in colonies.
Devices in these networks have limited communication capabilities, in that they can only emit a pulse of energy or sense the channel for a pulse of energy. 
Works on beeping networks assume perfect transmission and reception of messages, and the inter-node interference is the sole impairment. In reality, no receiver is perfect and is typically impacted by noise.  

We introduce \emph{noisy} beeping networks, where the observed communication is noisy with some fixed probability of error. Such noisy networks have implications for ultra-lightweight power limited sensor networks as well as biological systems. We show how to compute tasks in a noise-resilient manner assuming a noisy beeping network of arbitrary structure. 
In particular, we transform any  
algorithm over  a noiseless beeping network (of size~$n$) into a noise-resilient version while incurring a multiplicative overhead of essentially~$O(\log n)$ in its round complexity, with high probability. 
We show that our coding is optimal for some (short) tasks, such as the node-coloring of cliques. 
Interestingly, in the case of coloring, our technique achieves the same complexity as the standard beeping model, while being noise resilient.

We further show how to simulate a large family of algorithms designed for distributed networks in the CONGEST($B$) model over a noisy beeping network. 
The simulation succeeds with high probability and incurs an asymptotic multiplicative overhead of $O(B\cdot \Delta \cdot \min(n,\Delta^2))$ in the round complexity, where $\Delta$ is the maximum degree of the network. The overhead is tight for certain graphs, e.g., a clique. Surprisingly, our technique implies a \emph{constant overhead} coding for constant-degree networks. 

\end{abstract}


\section{Introduction}

In their \emph{Science} paper, 
Afek et al.~\cite{AABHBB11} 
used the beeping network model (also defined and formalized in~\cite{CK10})
to analyze 
distributed computation performed by multiple extremely simple agents 
such as cells in the fly's nervous system or ants in colonies~\cite{fonio2016}.
Studying this simplified communication model allows for better understanding of distributed computation problems where the computation is limited by communication. 
In this model, nodes are assumed to be able either to send a pulse of energy (`beep') or to sense an incoming transmission (`listen'), but they cannot do both at the same time. Furthermore, a node set to listen cannot tell how many of its neighboring nodes are beeping simultaneously, since all the node can sense is the existence of a pulse of energy, a single `beep'.

In this work we focus on a novel variant of the beeping model we call the \emph{noisy} beeping model.
This setting is parametrized by a noise parameter $\eps \in (0,1/2)$,  denoted  by~$BL_\eps$.
In this setting we assume that nodes are imperfect devices: sometimes they err in sensing the signal and deciding whether or not a beep was transmitted. Specifically, if a node is set to listen and none of its neighbors beep, the node still has a probability of $\eps$ to (falsely) hear a beep, independently of other nodes.
Additionally, if one or more of its neighbors beep, the node will hear nothing with probability~$\eps$, again, independently of the other nodes. 
Beeping nodes behave the same as they do in the standard $BL$~model.

The presence of noise invalidates almost all existing algorithms. 
To take a simple example consider the case    
where the goal of the nodes is to form an MIS of the network's graph
(see Section~\ref{sec:mis} for more details).
One way to find an MIS is the following: nodes randomly pick a number (say of size $\Theta(\log n)$ bits, where $n$ is the number of nodes in the network) and check who has the largest number in each neighborhood; this party joins the MIS.
To achieve this goal, each node can ``beep'' its chosen number: 
the node looks at the binary representation of its number and
for each of the next $\Theta(\log n)$~rounds, the node beeps if the corresponding bit of the binary representation of the number is ‘1’ and listens if it is ‘0’.
A node that is set to listen and hears a beep learns that it has a neighbor with a higher number. 
If the node does not learn about a neighbor with a higher number, it joins the MIS and beeps to signal this event; this node and its neighbors (who hear this beep) then finalize their output and terminate running the algorithm. Repeating the above enough times will yield an MIS.
It is easy to verify that a noisy beep can falsify the computation by 
causing two neighboring nodes to believe they have the highest number in their
neighborhood, or by causing a node and its entire neighborhood to quit without any of them ever joining the MIS.

Before describing our results, we discuss the noise model in more detail. 
The typical noise model in digital communication research papers
since Shannon's fundamental work
is additive \emph{receiver noise}, 
where each receiver experiences independent noise~\cite{gallager88,goldsmith2005wireless}.
Similar to the foundations laid by Shannon and the noisy model used by Gallager~\cite{gallager88}, our noise model assumes independent \emph{receiver noise} as well. 
That is, the receiving device features  faulty behavior, such that sometimes it fails to detect energy whereas at others it detects energy even though no other node has beeped. 
This behavior is typically caused  by amplifier (and other electronic) noise which is a fundamental limit of all physical receiving devices.
This aligns well with the fact that beeping networks commonly consist of extremely cheap, ultra-lightweight devices, where the receiver is merely a simple carrier sensing device which is prone to false-alarms and misdetection errors. 
Receiver noise also makes sense in biological systems such as ant colonies, where wind or other environmental disruptions may displace pheromones so that a nearby ant does not ``hear'' them, while a different ant located farther away does.

Another possible abstraction of the wireless channel noise assumes that each link between every transmitter and any receiver  adds independent noise~\cite{EKS20}. 
We argue that 
this kind of noise makes little sense in wireless networks and is too strong for arbitrary beeping networks. 
To better understand this, consider the case of a star network where a center node is connected to $n$ other nodes. If we allow \emph{channel} noise, then the center will constantly hear beeps: even when all the $n$~nodes in its neighborhood are set to listen, the center will hear a beep with probability $1-(1-\epsilon)^n$. %
Moreover, this probability tends 
towards~1
as more devices are present, \emph{even if they are all silent}. This makes little sense in the case of wireless networks since the noise probability should not depend on the number of devices, unless we assume devices are faulty and falsely transmit energy with probability~$\eps$ (this is also known as \emph{sender noise}).

\subsection{Results and Techniques}

\subsubsection{Noise-resilient beeping simulation}
We show how to perform any task that takes $R$ rounds assuming a noiseless beeping network
on a \emph{noisy} beeping network, with a multiplicative overhead of  $O(\log n +\log R)$ 
in the round complexity ($n$ is the number of nodes) with high probability, i.e., with a polynomially small failure probability. 
We further argue that this overhead
is tight for general simulations in the regime of $R=n^{O(1)}$, as there is a task that has a lower bound of $\Omega(\log n)$ in its overhead, compared to the noiseless task.

\begin{theorem}[main, informal]
\label{thm:main}
Given any network of $n$ nodes in an arbitrary topology and given 
any protocol $\pi$ of length~$R$ rounds in the beeping model (with or without collision detection),
$\pi$ can be simulated in a \textbf{noisy} beeping network (without collision detection)
in $R \cdot O(\log n+\log R)$ rounds with a probability of at least $1-2^{-\Omega(\log n+\log R)}$.
\end{theorem}

In the above, simulating a protocol $\pi$ (designed for model~$A$) in model~$B$ means that when running the simulation in model~$B$, the parties end up with a valid transcript of some instance of running $\pi$ in model $A$. Our theorem applies to both deterministic and random protocols. See further details in Section~\ref{sec:prelim}.

The simulation result of Theorem~\ref{thm:main} leads to exciting outcomes. For instance, it gives fast algorithms for computing coloring, MIS, and leader election over noisy beeping networks (Section~\ref{sec:applicaitons}).
In particular, we show how to perform coloring in $O(\Delta \log n + \log^2n)$ rounds, where $\Delta$~is the maximum degree of the network. 
We obtain this result by emulating 
the coloring protocol of Casteigts et al.~\cite{CMRZ19} via the method of Theorem~\ref{thm:main}. 
Surprisingly, Chlebus et al.~\cite{CMT17} prove a lower bound of $\Omega(n \log n)$ rounds for coloring\footnote{Note that the lower bound in~\cite{CMT17} regards coloring a clique with $n$ colors, while the scheme in~\cite{CMRZ19} uses $O(\Delta+\log n)$ colors, which is an easier task. However, given an $O(\Delta+\log n)$-coloring of the clique, one can perform a standard color reduction in~$O(\Delta+\log n)=O(n)$ rounds which yields an $n$-coloring of the clique.} a clique of size~$n$ \emph{in the $BL$ model}. 
This implies that our noise-resilient simulation is tight in this case 
since the $BL_\eps$ model is weaker than the~$BL$ model.

\medskip

\subsubsection{Noise-resilient collision detection}
Several relaxed versions of the beeping model appeared in the literature. 
These variants differ in the way they regard the event of a \emph{collision}---where more than a single node is beeping at a given time. Collision detection capabilities can be attributed to listening nodes, beeping nodes, or both.
Beeping-with-collision-detection means that if a node is set to beep, it can tell whether or not one or more of its neighbors also beeped on that same round. Listening-with-collision-detection means that a node set to listen can distinguish the following events: no neighbors beep; a single neighbor beeps; multiple neighbors beep. 
These models are denoted $B_{cd}L$, $BL_{cd}$, and $B_{cd}L_{cd}$, respectively, where the standard beeping model without any collision detection capabilities is denoted~$BL$.
In the last ten years, beeping networks have attracted considerable interest.
Previous works (e.g., \cite{AABHBB11,AABCHK13,GH13,JSX16,BBDK18,DBB18,CD19le,CMRZ19}) have 
considered all these four variants and have developed various algorithms that solve specific tasks, such as vertex coloring, leader election, or finding a Maximal Independent Set (MIS). 
In the noiseless beeping model, some tasks have different complexities when executed over channels with and without collision detection.
In many cases, the two settings are separated by a factor of $O(\log n)$ in their achievable round complexity. 
For instance, coloring a graph with a known maximum degree~$\Delta$ takes $O(\Delta \log n)$ rounds in the no-collision-detection $BL$ model~\cite{CK10} but only  $O(\Delta+\log n)$ rounds with collision detection, i.e., in the $B_{cd}L$ setting~\cite{CMRZ19}; finding an MIS takes $O(\log n)$ rounds if collision detection is present~\cite{JSX16,CMRZ19}, and up to $O(\log^2 n)$ rounds when no collision detection is available~\cite{AABHBB11,HL16}. %

\smallskip 

Our main technique in Theorem~\ref{thm:main} 
is a noise-resilient \emph{collision detection} method which reduces the 
$BL_\eps$ model to the stronger $B_{cd}L_{cd}$ model. 
That is, by executing this method, nodes can tell with high probability whether none, one, or more than one of their neighboring nodes  has beeped.
The cost of each instance of the collision detection procedure is $O(\log n)$~rounds, and it works with high probability in the presence of noise. 
In a sense, by blowing up the round complexity by a $O(\log n)$ factor, we achieve two goals: noise resilience and collision detection. 
While each goal might be trivial on its own assuming a blowup of $O(\log n)$, we achieve both by paying the overhead once.
This allows us to simulate any protocol designed for the $B_{cd}L_{cd}$ model (or any other weaker variant) over the noisy with no-collision-detection $BL_\eps$ network with high probability, without suffering a $O(\log^2 n)$ overhead.

As a consequence, the simulation stated in
Theorem~\ref{thm:main} can be applied to protocols defined in the stronger $B_{cd}L_{cd}$ model, which, as noted above, usually feature round complexity which is smaller by a factor $\Theta(\log n)$ to begin with.
This leads to the surprising result where for many tasks we ``pay no price'' for making them noise resilient, compared to the corresponding $BL$ algorithms (as both the $BL$ and $BL_\eps$ models do not feature collision detection, this is the right comparison).
In the case of coloring, this disparity in the model leads to tight results. 
Unfortunately, for other tasks such as MIS or leader election, there is still a gap between the lower and upper bounds in the noisy model (similar gaps also exist in the noiseless beeping setting). We elaborate on this in Section~\ref{sec:applicaitons}.

Note that for short protocols, where $R= \poly(n)$, 
the above reduction
suffices to show that all rounds succeed with high probability via a union bound. 
However, for longer protocols, there will still be errors every $\poly(n)$ rounds, which we can avoid with high probability by increasing the overhead of the collision detection  to~$O(\log R)$. 
To the best of our knowledge, all the beeping protocols that have appeared so far in the literature are of polynomial length, i.e., all the relevant protocols are indeed short.

\medskip

The idea behind the collision detection mechanism is very simple and straightforward. 
Assume we want to simulate a single round of the $B_{cd}L_{cd}$ model. Consider one specific neighborhood and call all the parties that want to beep in this round \emph{active}; denote the parties that listen as \emph{passive}.
The parties pick a random codeword from a binary code~$C$ of length $\Theta(\log n)$ that has both \emph{good distance} and \emph{constant weight}. With high probability each one of the nodes in a given neighborhood picks a unique codeword from~$C$.
Next, each active party beeps its codeword: for the next $i=1$ to $\Theta(\log n)$ rounds, 
it beeps if and only if the $i$-th bit of its codeword is~1. See Figure~\ref{fig:CD} for a demonstration.

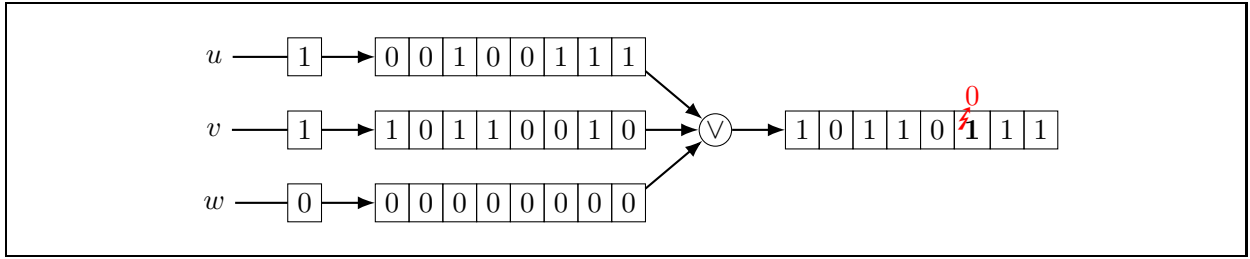
\begin{figure}[ht]
\smallskip
\begin{framed}
\centering
\begin{tikzpicture}[
node distance=0pt,
start chain =A going right,
start chain =B going right,
start chain =C going right,
start chain = L going right,
    X/.style = {rectangle, draw,%
                minimum width=2ex, minimum height=3ex,
                outer sep=0pt}
]

\matrix[row sep=0mm,column sep=7mm] {
\node (u) {$u$}; & \node[X] (u1){1}; &  
	\coordinate (u3);
	\foreach \i [count=\xi] in {0,0,1,0,0,1,1,1 }
	     \node[on chain=A,X] (A\xi) {\i};
& & \\
\node (v) {$ v$}; & \node[X] (v1) {1}; & 
	  \foreach \i  [count=\xi] in {1,0,1,1,0,0,1,0 }
	    \node[on chain=B,X] (B\xi) {\i};    
&  \node(v4) [circle,draw,inner sep=1pt] {$\vee$} ;
&  	\foreach \i [count=\xi] in {1,0,1,1,0,\textbf{1},1,1 }
	     \node[on chain=L,X] (L\xi) {\i};
	     \node at (L6) [red,below right=5pt,rotate=180] {\large\Lightning};
	     \node at (L6) [red,above=2mm] {$0$};

\\
\node (w) {$w$}; & \node[X] (w1) {0}; & 
	  \foreach \i [count=\xi] in {0,0,0,0,0,0,0,0 }
	    \node[on chain=C,X] (C\xi) {\i};	
	    \node at (C6) [white,above=2mm] {\phantom{$0$}}; %
& & \\
  };

 \draw  [thick,-Latex] (u) -- (u1) -- (A1);
  \draw  [thick,-Latex] (v) -- (v1) -- (B1);
   \draw  [thick,-Latex] (w) -- (w1) -- (C1);
   
 \draw [thick,-Latex] (A8) -- (v4);
  \draw [thick,-Latex] (B8) -- (v4);
  \draw [thick,-Latex] (C8) -- (v4);  

\draw [thick,-Latex] (v4) -- (L1);

\end{tikzpicture}
\end{framed}
\vspace{-0.75em}
\caption{\footnotesize A demonstration of the collision detection scenario. Nodes $u$ and~$v$ are active, and $w$ is passive. Each of the active parties picks a random codeword from a code of weight~$4$, and beeps it. The channel superimposes the beeps. The weight of the received superimposed  transmission indicates whether there were no active parties, one active party, or whether a collision has happened. Noise may flip some of the bits a certain party hears (in this example, the 6th bit was flipped by the noise for some receiver). }
\label{fig:CD}
\end{figure}

The three cases (no transmission, single transmitter or multiple transmitters) can be distinguished by the constant weight property of the code. 
If no party is active, then no party beeps (the ``all zero'' word is transmitted). If a single party beeps, then a single codeword is beeped, and its expected weight matches the weight of the code. 
In the case where multiple parties beep, the expected weight of the transmission is much higher (this stems from the distance property of the code). 
Each party, active or passive, counts the number of beeps and decides between the three cases accordingly.
The noise causes each party to hear a noisy version of the transmissions; however, the expected number of beeps is determined solely by the three cases above. 
Then, our analysis shows that 
the probability that the noise corrupts enough beeps to yield a wrong outcome is polynomially small in~$n$, given the right choice of code parameters.

In terms of lower bounds, 
it is not too difficult to see that collision detection over noisy networks requires~$\Omega(\log n)$ rounds to succeed with high probability.
The noise has a polynomial probability to corrupt $O(\log n)$ slots, which means it can invalidate any protocol of that length. 
Therefore, our collision detection protocol is optimal.
\begin{theorem}[collision detection, informal]
The task of Collision Detection over noisy beeping networks ($BL_\eps$) with a polynomially small failure probability
requires $\Theta(\log n)$ rounds.
\end{theorem}

\subsubsection{Noise-resilient simulation of CONGEST protocols}

Our previous result considers the message-passing model, i.e., the CONGEST model with messages of size~$B$ bits.
We show that any fully-utilized\footnote{Fully-utilized protocols are ones in which in every round every party sends a message to each one of its neighbors.} protocol~$\pi$ in the CONGEST($B$) model can be simulated over the noisy beeping network~$BL_\eps$ with high probability, with a multiplicative overhead of $O(B\cdot \Delta \cdot \min(n,\Delta^2) )$ in the round complexity (as the length of the protocol, $|\pi|$, tends to infinity). 
Note that for networks with a constant degree our simulation imposes a constant overhead when~$B=O(1)$. 

The simulation operates in two steps, similar to the approach in~\cite{BBDK18,ABLP89}.
In the first step the parties perform 2-hop coloring. After this step,
each node is assigned a ``color'' (a number) so that no two nodes of distance at most two share the same color.
In the second step the coloring is used to obtain Time-Division Multiple Access (TDMA). 
That is, assume rounds are also colored in a sequential cyclic manner and let nodes only speak in rounds associated with their color.
Since we use a 2-hop coloring, it is guaranteed that only a single node beeps in every neighborhood, so no collisions occur. When a node gains exclusive access to the channel, it transmits messages to all its neighbors at once. 
This TDMA blows up the number of rounds by the number of colors times the maximum degree of the network; note that in a 2-hop coloring, the number of colors is bounded by~$\min(\Delta^2,n)$. 

On top of the TDMA we employ a coding technique for interactive coding over noisy networks developed by Rajagopalan and Schulman~\cite{RS94,ABEGH19}. 
This coding allows the parties to replace~$\pi$ 
with a \emph{resilient} protocol that computes the same functionality as~$\pi$
even if some of its messages arrived corrupted. 
This coding requires a bound on the per-message noise probability to be at most~$O(\Delta^{-1})$, which naively implies an additional blowup of $O(\log \Delta)$ in order to reduce the noise to that level. 
We avoid this blowup by taking advantage of the broadcast nature of the beeping model. 
Each node can concatenate the messages directed to \emph{all} its neighbors into a single message (of length~$O(B\Delta)$) that will be encoded against errors via a standard error-correcting code with a constant distance. 
The node beeps this single message, and all its neighbors will hear and decode it simultaneously.
This method reduces the per-message noise level to $2^{-\Omega(\Delta)}$ with only a \emph{constant} overhead.

The following theorem summarizes this result.
\begin{theorem}[Simulating message-passing protocols, informal]
\label{thm:MPsim-inf}
Given a 2-hop coloring with $c$~colors, 
any fully-utilized protocol~$\pi$ 
in the CONGEST(B) model can be simulated 
with high probability in the noisy beeping model $BL_\eps$ 
in $O(c^2\log n) + \max(|\pi|,\log n/\Delta)\cdot O(B\cdot c \cdot \Delta)$~rounds.
\end{theorem}

In~\cite{BBDK18}, Beauquier et al.\@ showed how to simulate a CONGEST(B) protocols over $BL$ networks with $O(B\cdot c^2)$~multiplicative overhead. 
Hence our simulation (Theorem~\ref{thm:MPsim-inf}) improves the result of~\cite{BBDK18} for some networks (e.g., when $\Delta \ll n$), in addition to being noise-resilient.

Finally, we show that our simulation is tight in the case of a clique. We show a  task (actually, a family of tasks, parametrized by~$k$) that can be computed in $k$ rounds in the CONGEST(1) model, but takes $\Omega(kn^2)$ rounds in the $BL$ (or $BL_\eps$) model over a clique with $n$~nodes. 
Note that our simulation's multiplicative overhead over a clique, as given by Theorem~\ref{thm:MPsim-inf},
 is exactly~$\Theta(n^2)$ since the length of $\pi$ tends to infinity ($k\to\infty$). 

\subsection{Related Work}
The beeping model was introduced by Cornejo and Kuhn~\cite{CK10}, and was
followed by numerous works targeting specific  tasks in the four variants of the beeping model (with and without collision detection),
for both single-hop (clique) and multi-hop networks. 
These include algorithms for  finding an MIS~\cite{AABHBB11,AABCHK13,JSX16,HL16,BBDK18,CMRZ19},
(node) coloring of the network~\cite{CK10,JSX16,CMT17,CMRZ19}, leader-election~\cite{GH13,FSW14,DBB18,CD19le}, self-stabilization and synchronization~\cite{GM15},
and
broadcasting a message~\cite{CD19bc}. 
Simulating beeping algorithms with collision-detection ($B_{cd}L_{cd}$, $B_{cd}L$, and $BL_{cd}$) 
over beeping networks without collision-detection $BL$ was shown by~\cite{CMRZ19counting,CMRZ19}; the simulation incurs an $O(\log n)$ factor.

Hounkanli, Miller, and Pelc~\cite{HMP20} considered the case of a \emph{single-hop, noisy} beeping channel, where the noise can only make beeps disappear. They provide algorithms for global synchronization and consensus in this setting.
Efremeko, Kol, and Saxena~\cite{EKS20}
considered lower bounds on \emph{single-hop, noisy} beeping channels. In contrast to our work, 
in their noise model  the 
(single) channel can convert a beep to silence and vice-versa, but this affects \emph{all} the listening parties.
They showed that a multiplicative overhead of $\Omega(\log n)$ is necessary for coding
beeping protocols against the above noise in the single-hop setting.

\medskip

\emph{Radio Networks}~\cite{CK85} are closely related to beeping networks. In this setting, wireless devices
can communicate with their neighbors by sending messages of some fixed size; however, if a collision occurs (more than a single sender on a given round), no message is delivered. While this model seems very similar to the beeping model, there are significant differences that stem from the fact that collisions in the beeping model superimpose while in the radio network model they destructively interfere.

For example, consider the case of broadcasting a message. This task can be done in $O(D+M)$~rounds in the BL model (where $D$ is the diameter of the network, and $M$ the length of the message), by employing ``beep waves''~\cite{GH13,CD19bc}.
On the other hand, in the radio networks model there are graphs with diameter~$D=2$ that require $\Omega(\log^2n)$ to broadcast a message~\cite{ABLP91}.
Kushilevitz and Mansour~\cite{KM98jur} 
proved a lower bound of $\Omega(D\log \frac{n}{D})$ for graphs with diameter~$D$.
Similar to the case of beeping,
there are variants with and without collision detection. Bar-Yehuda, Goldreich, and Itai~\cite{BGI91} showed how to simulate a single-hop radio network with collision detection over a multi-hop network without collision detection with $O(( D+\log\frac{n}{\delta})\log\Delta)$ overhead ($1-\delta$  was the desired success probability of their scheme). This simulation led to an almost optimal broadcast with $O(D\log n + \log^2n)$. Later, Czumaj and Rytter~\cite{CR06}, as well as Kowalski and Pelc~\cite{KP05}, obtained optimal broadcast algorithms (assuming no spontaneous transmissions) with $O(D\log\frac{n}{D}+\log^2 n)$ rounds.

The case of \emph{noisy} radio networks  was introduced by Censor-Hillel et al.~\cite{CHHZ17}. They provided a noise-resilient broadcast algorithm with $O(\frac{\log n}{1-\eps}(D+\log \frac{n}{\delta}))$ that succeeded with probability $1-\delta$ over a radio network where each transmission is be corrupted with probability~$\eps$. Further, Censor-Hillel et al.~\cite{CHHZ19} showed how to simulate radio network protocols over noisy radio networks with high probability, with an overhead $\poly(\log  \Delta, \log\log n)$ for non-adaptive protocols, and $O(\Delta \log^2 \Delta)$ for any protocol. Lower bounds of $\Omega(\log n)$ on the overhead required for interactive coding for noisy (erasure) radio networks were presented by Efremenko, Kol and Saxena~\cite{EKS19}.

\medskip 
Another closely related field is \emph{interactive coding}~\cite{gelles17}, where computations are performed  over noisy networks (of various kinds) in the message-passing model,  and the aim is to perform the computation in a noise-resilient manner with small overhead. 
Many works have considered the case of computation over noisy networks, and has given coding schemes assuming random noise~\cite{RS94,GMS14,ABEGH19}, coding schemes for worst-case noise~\cite{JKL15,HS16,CGH19,GKR19,ADHS19,EKS18,MG21},
and lower bounds on the overhead assuming random noisy channels~\cite{BEGH18,GK19}.

\section{Preliminaries}\label{sec:prelim}

\paragraph{Standard Notations.}
For a positive integer $n\in \mathbb{N}$ we denote by~$[n]$ the set $\{1,\ldots,n\}$. We use standard Bachmann-Landau notations, e.g.\@ $O(\log n)$ and $\Omega(\log n)$, to denote the asymptotic behavior of a function as $n$ (or a different parameter, usually explicitly written) tends to infinity. %
All logarithms are taken to base~2.

\paragraph{The Beeping Communication Model.}
Throughout this work we assume a network described by an undirected graph $G=(V,E)$ with $n=|V|$ nodes and $m=|E|$ edges. Each node $u\in V$ is a party that participates in the computation. Edges in~$E$ represent pairs of neighbors that can hear each other. 
For a given node~$v$, the set
$\neigh_v = \{ u \mid (u,v) \in E\}$ is called the neighborhood of~$v$. We set $\neigh_v^+ = \{v\}\cup \neigh_v$ to be the neighborhood of~$v$ including itself.
The \emph{maximum degree} of the network, $\Delta = \max_{v\in V} |\neigh_v|$, is the size of the largest neighborhood in the network. The diameter~$D$ of the network is the maximal length of the shortest path between any two nodes. For any $n \in \mathbb{N}$, we denote by~$K_n$ the clique of size~$n$ (also known as a \emph{single hop} network of $n$ parties; otherwise, this is a \emph{multi hop} network).

We assume that the size of the network, $n$, is known to all nodes; however, the topology of the network is unknown. Moreover, the nodes are assumed to be identical, that is, they run the same algorithm and have no distinguishing identifier. However, we assume that each node has its own stream of independent random bits, hence, the actions of the nodes might differ due to their difference in randomness.

Communication over the network is performed in synchronous \emph{slots} 
where for each slot each party can
either \emph{beep} or \emph{listen}.
If some party~$v$ beeps at a given slot, 
then any of its neighbors $u \in\neigh_v$ set to \emph{listen}, hears a beep. 
Parties that are set to listen in a given slot either hear a \emph{beep},
which implies that at least one of their neighbors beeped in that slot, 
or they hear \emph{silence}, which means that none of their neighbors beeped. 

We consider four possible models that differ in the collision-detection capabilities of the parties.
In the $BL$ model (no collision detection for beeping or listening nodes), in any slot where $v$ beeps, it cannot tell if any of its neighbors beeped or not. Further, in any slot where $v$~listens and hears a beep, it cannot tell how many of its neighbors beeped at that slot.

We can relax the $BL$ model and give collision-detection capabilities to either beeping nodes (denoted with $B_{cd}$) or to listening nodes (denoted with $L_{cd}$).
In the $B_{cd}L$ and the $B_{cd}L_{cd}$ models, parties that beep at a given slot can distinguish between the event that none of their neighbors beeped at that slot and the event where one or more of their neighbors beeped.
In the $BL_{cd}$ and $B_{cd}L_{cd}$ models, a party that listens and hears a beep can distinguish between the event where a single neighbor beeped and the event where multiple neighbors beeped.

\paragraph{Noisy Beeping Model.}
In the noisy beeping model, $BL_\eps$, the communication works similarly to the $BL$ model defined in the paragraph above, except that random noise might alter the event of hearing a beep to a silence, and vice-versa.
In particular, 
for any party~$v$ set to \emph{listen} at a given slot, 
if the anticipated outcome of the node (in the noiseless $BL$ model) in this specific slot
is $out \in \{beep, silence\}$, then with probability~$\eps$ its outcome in the noisy model will be the other possible outcome. 
We assume that the noise crossover probability 
is in the range of $\eps\in(0,1/2)$ and  that the noise
is independent between different nodes and across different time-slots of the same node.
We further assume $\eps$ is known to all the parties.

\paragraph{Protocols.}
Nodes run \emph{a protocol}---a distributed algorithm that aims to solve some distributed task (e.g., to color the network, elect a leader, or broadcast a message).
The protocol dictates to each node when to beep or listen and what output to give, as a function of the current time slot, all previous communications, and the node's input and randomness. 
The \emph{transcript} of a given party is the sequence of sent and received messages (or beeps) observed by that party during the execution of the protocol. The protocol's transcript is the aggregation of all the parties' transcripts.

We assume that a protocol $\pi$ can run on any graph~$G$ with $n$ nodes, and denote an instance of such execution by~$\pi(G)$; we usually omit~$G$ when it is clear from context. (We sometimes restrict $\pi$ to a certain family of graphs, e.g., cliques~$K_n$.)
The \emph{length} or  \emph{round complexity} of a protocol~$\pi$ on~$G$, denoted~$|\pi(G)|$, is the (maximal) number of rounds (slots) it takes until all parties have terminated.
We naturally extend this notion to denote the asymptotic behavior of~$\pi$ over a family of graphs $\{G_i = (V_i,E_i)\}$, with $|V_i|\to\infty$. 

The beeping protocols discussed in this work are randomized Monte-Carlo protocols, 
where the computation succeeds \emph{with high probability}. That is, with  a probability of at least $1-n^{-\Omega(1)}$, on any sequence of graphs $G_n$ as $n\to \infty$. 
To facilitate the discussion, we assume that the length of~$\pi$ (on a given~$G$) can only depend on~$G$; in particular, it is independent of the parties' randomness. 
Our results easily extend to the case where the length is a random variable. 
Note that, in contrast to the above notion, some works (e.g.,~\cite{CMT17,CD19le}) define $|\pi|$ as the expected computation time, which significantly affects the lower and upper bounds they obtain.
Finally, bear in mind that any protocol designed to succeed with high probability in $BL_\eps$ for a specific~$\eps<1/2$, will also succeed with high probability in $BL_{\eps'}$ for  any $0\le \eps' < \eps$. Further, recall that by repeating each transmission $m$ times and taking their majority, 
one can reduce $BL_\eps$ to $BL_{\eps'}$. 
When $\eps,\eps'$ are both constants, then $m=O(1)$, leading to the same asymptotical complexity.

\paragraph{Simulating Protocols.}
We say that a protocol~$\Pi$ (in some given communication and noise model)
\emph{simulates}~$\pi$ (possibly defined over a different model), 
if after the completion of~$\Pi$, each node outputs the transcript it would have seen when running~$\pi$ assuming noiseless channels.   
To be more exact, let $\pi(G,rand)$ be the instance of $\pi$ on $G$ where the parties get the randomness~$rand$. Then a simulation $\Pi(G,rand, rand')$ of~$\pi$ gives the parties the randomness $rand$ and $rand'$ and aims to have the parties output $\pi(G,rand)$ with high probability \emph{over the randomness~$rand'$}.

In Section~\ref{sec:mp} we discuss a slightly different type of beeping algorithm, whose purpose is to simulate a CONGEST protocol over a beeping network. 
In fact, we require these simulation protocols to succeed with high probability with two complementary notions: 
(1) when $\pi$ is defined over any~$G$, and then $\Pi$ is required to succeed with probability~$1-n^{-\Omega(1)}$ as $n\to\infty$, where $n$ is the size of~$G$, and
(2) when $\pi$ is an infinite family of increasing-length protocols that all are defined over \emph{the same network~$G$}. Then, $\Pi$ needs to succeed with probability~$1-|\pi|^{-\Omega(1)}$ as $|\pi|\to\infty$. Our simulations in Section~\ref{sec:mp} satisfy both these notions; the simulations in the rest of this paper adhere to notion~(1).

The \emph{overhead}  (in rounds) of simulating $\pi$ describes the increase in the round complexity of~$\Pi$ as a function of~$|\pi|$. When $|\Pi| = f |\pi| + g$ for some functions $f,g$, the term $f$ is the multiplicative overhead and $g$ is the additive overhead. We focus on reducing the multiplicative overhead. We write the overhead using asymptotic notations $O()$, $\Omega()$ etc., where the asymptotic behavior holds either when $n\to\infty$ (e.g., when the algorithm $\pi$ is fixed and $G$ changes), or when $|\pi|\to\infty$ (e.g., when $G$ is fixed and the protocol $\pi$ to be simulated changes).

\paragraph{Error Correcting Codes.}
A \emph{code} is a mapping $C: \Sigma^{rn} \to \Sigma^n$, where $\Sigma$ is some alphabet. We call 
$r\in (0,1)$ the {\em rate} of the code and $n\in\mathbb{N}$ the {\em block length} of the code. 
The elements in the image of $C$ are called {\em codewords}. 
With a slight abuse of notation, we use~$C$ to also denote the set of codewords.

 A binary code is one with~$\Sigma=\{0,1\}$. 
We say that $C$ has {\em relative distance} $\delta\in[0,1]$, if for
every pair of distinct vectors $x,y \in C$ it holds that
$\Hamming(x,y)\ge\delta n$, where $\Hamming(\cdot,\cdot)$ is the Hamming distance function. The weight of a codeword~$x\in C$, denoted $\Weight(x)$, is the number of non-zero indices in~$x$, i.e., its Hamming distance from the all-zero word.

It is well known that binary codes with constant rate and relative distance exist.
Such codes can be obtained, for example, by concatenating Reed-Solomon codes~\cite{RS60} with binary linear Gilbert-Varshamov codes~\cite{G52,V57}.
\begin{lemma}[{\cite[Theorem~1]{justesen72}}]
\label{lem:binary-codes}
For any given rate $\rho<1/2$ and for any $m>0$, there exists a binary code $C_m: \{0,1\}^{\rho_m n}\to\{0,1\}^n$ with $n=2m(2^m-1)$ such that the code $C_m$ has rate $\rho_m \ge \rho$ and relative distance $\delta_m > 
(1-2\rho)H^{-1}(\tfrac12)$.
\end{lemma}
In the above, $H(x)=x\log(1/x)+(1-x)\log(1/(1-x))$ is the binary entropy function. Note as well that the above codes can be encoded and decoded in polynomial time in~$n$~\cite{justesen72}.

\paragraph{Chernoff's Inequality.}
We will use the following Chernoff bound~\cite{chernoff1952}.
\begin{lemma}[Corollary~4.6 in~\cite{MUbook}]
\label{lem:chernoff}
Let $X_1,\ldots, X_n$ be independent $\{0,1\}$-random variables such that $\Pr[X_i=1]=p_i$. Let $X=\sum_{i=1}^n X_i$ and $\mu = E[X]$. For $0<\delta<1$, $\Pr[|X-\mu|\ge \delta \mu]\le 2e^{-\mu\delta^2/3}$.
\end{lemma}

\section{Collision Detection in the Noisy Beeping Model}

The \emph{collision-detection} task is defined as follows. Each party is either \emph{active} or \emph{passive}. At the end of the scheme each party needs to output whether (i) all the nodes in its neighborhood (including itself) are passive, (ii) there exists exactly a single active node, or (iii) more than one node is active. 

In this section we develop a collision-detection algorithm for beeping networks, 
which succeeds with high probability even in the presence of noise, i.e., in the $BL_\eps$ model.
The basic idea is very simple: we use a (binary) constant-weight code~$C$ where every codeword has the same number of~$1$'s. Parties beep the binary representation of codewords from~$C$ (this means they beep for each 1 of the codeword). 
Since the code is of constant weight, if exactly a single party is active, then the number of times it beeps is exactly the code's weight. 
If the number of beeps heard is substantially larger than that, we can infer that more than a single party was active. 
If on the contrary the number of beeps is substantially smaller,
then we infer that no party was active. Note that noise may add or remove beeps, and cause the parties to infer the wrong outcome. However, this event occurs with a  polynomially small probability due to our choice of parameters.

In particular, we use a \emph{balanced} binary code~$C$ of length~$n_c$.
A balanced binary code is one where the Hamming weight of any codeword is exactly~$n_c/2$. 
The code~$C$ has a constant relative distance $\delta$ and a constant rate~$r$. 
To be concrete, we can construct~$C$ by taking any binary code with a constant relative distance and rate (Lemma~\ref{lem:binary-codes}) and concatenate it with a balanced code of size~2, e.g., $0\to 01$ and
$1\to 10$. This concatenation makes the code balanced while preserving its distance. 
The rate decreases by a constant factor of~$2$.

In any given instance of the collision-detection protocol, every party that wants to beep is ``active''.
Any other node is ``passive''. 
Each active node uniformly at random picks a codeword~$c\in C$ and for the next $i=1,2,\ldots,n_c$ rounds it beeps if $c_i=1$, or keeps silent otherwise.   

As mentioned above, if no party is active, then no node beeps during the $n_c$~rounds. Any beep that is heard by some node in this case must be associated with noise. In particular, the number of expected beeps each node hears is~$\eps n_c$.
If there is only a single active node, each node hears $n_c/2$ beeps in expectation.
However, when two or more nodes are active, with high probability they have chosen different codewords $c_1,c_2$. 
Since $\Hamming(c_1,c_2)>\delta n_c$, at least $n_c/2+(\delta/2) n_c$~beeps are sent, which means that at  least $n_c(1/2+ \delta/2 -\eps\delta)$~beeps are heard by each node in expectation (including its own beeps).

By choosing a large enough $\delta$ with respect to the noise probability~$\eps$, 
we are able to distinguish between these three cases. 
The following technical lemma lower bounds the number of beeps in a case of a collision.
That is, if two (or more) adjacent nodes try to beep different codewords $c_1, c_2 \in C$, 
then in at least $n_c/2 + \delta n_c/2$ rounds there exists some node that beeps.

\begin{claim}\label{clm:codewordOR}
Let $C$ be a balanced binary code with length~$n_c$ and relative distance~$\delta$.
For any two distinct codewords $c_1,c_2 \in C$, the Hamming weight of $c_1 \vee c_2$ (i.e., the bit-wise OR) is at least $n_c(1+\delta)/2$.
\end{claim}
\begin{proof}
In a balanced code $\Weight(c_1)=\Weight(c_2)=n_c/2$. Consider  $d=c_1 \oplus c_2$. It has at least $\delta n$ non-zero elements, which we denote~$I$. Let $c\in\{c_1,c_2\}$ be the codeword that has minimal weight on the support~$I$, and denote by $c_I$ the codeword~$c$ restricted to the indices in~$I$. 
Since $c_1$ and $c_2$ differ on each index in~$I$ and $c$ is the one with minimal weight on~$I$,
we get  that, $\Weight(c_{I})\le |I|/2$. 
Thus, the weight of $c_1 \vee c_2$ is at least the weight of $c$ plus the zeros in $c_{I}$, where the other codeword is 1,
\[
\Weight(c_1\vee c_2) \ge n_c/2 + (|I|-|I|/2) \ge n_c/2 + \delta n_c/2. \qedhere
\]
\end{proof}

The noise-resilient collision-detection protocol is depicted in Algorithm~\ref{alg:collision}. 
Each node begins with an input of  ``active'' or ``passive'' that determines if it wants to beep or just wants to detect if one or more neighbors want to beep.
\begin{algorithm}[ht] %
\caption{Collision detection over a noisy beeping model, $BL_\eps$}
\label{alg:collision}
\begin{algorithmic}[1]
\small
\Statex Assume a balanced binary code 
$C$ of length~$n_c$, relative distance $\delta<1/2$ and rate~$r$ (which may depend on~$\eps$).

\Statex

\Procedure {CollisionDetection}{active/passive}
	\If {passive}
		\State \emph{listen} for $n_c$ rounds  
	\ElsIf{active}
		\State Pick a codeword $c\in C$ uniformly at random.
		\For {round $j \in \{1,\ldots, n_c\}$}
			\State If $c_j=0$, \emph{listen}.
			\State If $c_j=1$, \emph{beep}. 
		\EndFor
	\EndIf

	\State Let $\chi$ be the number of beeps sent plus heard during the above $n_c$~rounds.
	\If {$\chi<\frac{n_c}{4}$}
		\State \Return $\silence$
	\ElsIf{ $\chi < %
			\tfrac12(1+\delta/2) n_c$}  
		\State \Return $\single$ 
	\Else
		\State \Return $\col$
	\EndIf

\EndProcedure

\end{algorithmic}
\end{algorithm}

We now formally analyze the correctness of the \textsf{CollisionDetection} procedure of Algorithm~\ref{alg:collision}. Namely, we prove that except with polynomially small error probability, any node outputs whether zero, one, or more than one of its neighbors were active.
We begin by considering a single node and proving it can detect collisions within its own neighborhood, with high probability. Then, we union bound over all the nodes in the network.
\begin{theorem}\label{thm:collisionDetection}
Assume $G$ is a network with~$n$ nodes in the $BL_\eps$ model, where all nodes perform 
the $\Call{CollisionDetection}{}$ procedure. 
Further assume that $n_c= \Omega(\log n)$, $\delta>4\eps$, and $r=\Theta(1)$.
Then, for any specific node $v$, the following three claims hold, with probability 
\(
1- n^{-(1+\Omega(1))}.
\)
\begin{enumerate}
\item If the number of active parties in~$\neigh_v^+$ is two or more, then $v$ outputs~$\col$.
\item If none of the parties in~$\neigh_v^+$ is active, then $v$ outputs~$\silence$. 
\item If the number of active parties in~$\neigh_v^+$ is exactly one, then $v$ outputs~$\single$.
\end{enumerate}
\end{theorem}

\begin{proof}
We prove the above three claims in turn.\\
\textit{Claim 1.}
Assume without loss of generality that $p_1,p_2$ are two of the active parties guaranteed by the theorem statement.  
Both~$p_1,p_2$ choose a random codeword from~$C$ and  with probability $1-1/2^{rn_c} = 1-2^{-\Omega(\log n)}$, their chosen codewords are not identical. 
Using Claim~\ref{clm:codewordOR} we know that out of the $n_c$~rounds, in at least $n(1+\delta)/2$ rounds some node in~$\neigh_v^+$ beeps. Note that if only a single node in~$\neigh_v^+$ is active, it beeps in only~$n_c/2$ slots.
We set $\alpha = (1+\delta/2)/2$ so that $\alpha n_c$ is the average of these two quantities; 
this will be the threshold of beeps that distinguishes between these two events.

The probability that $v$ counts fewer than $\alpha n_c$ beeps when two or more nodes are active in~$\neigh_v^+$, is less than the probability to have at least $(\delta/4)n_c$ noisy slots in either $n_c/2$ slots (if $v\in\{p_1,p_2\}$ is active) or across the entire $n_c$ slots (otherwise). Since the latter event is more probable, the error probability can be bounded by
\begin{equation}\label{eq:chr_delta|eps}
2e^{-\eps n_c\left(\frac{\delta/4-\eps}{\eps}\right)^2/3 } = e^{-\Omega(\log n)} = n^{-\Omega(1)}.
\end{equation}
This follows from Chernoff's inequality, Lemma~\ref{lem:chernoff},  assuming $\eps<\delta/4$.
The above is polynomially small. 
Note that the constant in the $\Omega(\cdot)$ in the exponent
can be made arbitrarily large by increasing the length of the code~$n_c$ by a constant factor (that may depend on $\eps,\delta$).
\\

\noindent\textit{Claim 2.}
Since none of the parties in~$\neigh_v^+$ is active, none of them beeps during the protocol. 
Then, the probability that
$v$ counts more than $n_c/2$ (noise) beeps is bounded by a polynomially small probability,
\begin{equation}\label{eq:chr_12eps}
2e^{-\eps n_c \left(\frac{1/2-\eps}{\eps}\right)^2/3}
= e^{-\Omega(\log n)}.
\end{equation}
Again, we can determine the constants by tweaking~$n_c$ so that $v$ outputs~$\silence$  with probability at least~$1-n^{-(1+\Omega(1))}$.

\noindent\textit{Claim 3.}
Since there is a single active node in~$\neigh_v^+$, this node makes exactly $n_c/2$  beeps, and none of the other nodes beeps. 
We split this into two cases. The first case is when the single active node is not~$v$. Then, 
since the noise is symmetric and the code~$C$ is balanced, the expected number of beeps $v$ counts is given by 
\(
n_c/2 \cdot (1-\eps) + n_c/2 \cdot \eps = n_c/2
\).

The node~$v$ will output~$\col$   
if due to the noise it counts more than $\alpha n_c = n_c/2 + (\delta/2)n_c/2$ beeps.
Again, using Lemma~\ref{lem:chernoff} we bound this probability by
\begin{equation}\label{eq:chr_alpha12}
2e^{-\frac{n_c}{2}(2\alpha-1)^2/3}= e^{-\Omega(\log n)}.
\end{equation}

The other case is when $v$ is the only active node in~$\neigh_v^+$. Then, $v$ counts its own $n_c/2$ with probability~$1$, and outputs~$\col$ if it hears at least $(\alpha-1/2)n_c/2= (\delta/2)n_c/2$  (noisy) beeps in the remaining $n_c/2$~slots. This erroneous event has a probability of at most
\(
e^{-\Omega(\log n)}
\), similar to Eq.~\eqref{eq:chr_delta|eps}.

Finally, note that $v$~outputs $\silence$ only if the total number of beeps it counts is below~$n_c/4$. This is impossible when $v$ is active. If a different node is the active one, counting so few beeps requires the number of noisy-slots to be at least $n_c/4 > (\delta/4)n_c$. This has a probability of at most $e^{-\Omega(\log n)}$ as argued above.
\end{proof}

Theorem~\ref{thm:collisionDetection} is stated for a single node $v$ and its neighborhood. 
Via a simple union bound, the same holds for any node in any network that performs  Algorithm~\ref{alg:collision}.
\begin{corollary}\label{cor:CD}
Assume $G$ is a network with~$n$ nodes in the $BL_\eps$ model, where all nodes perform 
the $\Call{CollisionDetection}{}$ procedure, with $n_c= \Omega(\log n)$, $\delta>4\eps$, and $r=\Theta(1)$. 
The three conditions of Theorem~\ref{thm:collisionDetection} hold for all nodes $v\in G$ with probability $1-n^{-\Omega(1)}$.
\end{corollary}

Next we claim that no collision detection protocol over~$K_n$ in $BL_\eps$ succeeds with high probability in $o(\log n)$ rounds.
\begin{lemma}\label{lem:col-lb}
Performing Collision Detection over~$K_n$ with high probability in the $BL_\eps$ model takes $\Omega(\log n)$ rounds.
\end{lemma}
\begin{proof}
Let $t$ be the length of the collision detection protocol and consider a specific party~$u$ and the rounds in which $u$ listens. Note that the pattern $u$ hears in these rounds determines its output, and that there always exist two different patterns that cause $u$ to output different outcomes, whether $u$ is active or passive.

Recall that for any given slot, the noise ``flips'' what $u$ hears with probability~$\eps$. It follows that the noise can cause~$u$ to output the wrong outcome with a probability of at least~$\eps^t$. Requiring $u$ to succeed with high probability implies $\eps^t < n^{-c}$ for some constant~$c$, hence, $t=\Omega(\log n)$.
\end{proof}
An immediate corollary of the above lemma is that detecting collisions on general (worst case) graphs in~$BL_\eps$ takes $\Omega(\log n)$ rounds to succeed with high probability.
Note that a similar lower bound of $\Omega(\log n)$ rounds for collision detection also holds in the~$BL$ model,
e.g., for $n/2$ pairs of nodes~\cite{AABCHK13}, or the wheel graph~\cite{CMRZ19}. These proofs also carry over to the $BL_\eps$~model. 

\smallskip

Corollary~\ref{cor:CD} and Lemma~\ref{lem:col-lb} lead to the following.
\begin{corollary}\label{cor:collsion}
Solving the Collision Detection task on~$n$ nodes with high probability in the noisy beeping $BL_\eps$ model has a round complexity of~$\Theta(\log n)$.
\end{corollary}

\section{Noise-Resilient Beeping Protocols from Noiseless Protocols}
\label{sec:simulaiton}
\subsection[Simulating B\_{cd}L\_{cd} protocol over BL\_e networks]{Simulating $B_{cd}L_{cd}$ protocol over $BL_\eps$ networks}

Given the \textsf{CollisionDetection} procedure of Algorithm~\ref{alg:collision}, one can simulate any beeping 
protocol that assumes that nodes can distinguish silence from a single beep and from
multiple beeps (i.e., the $B_{cd}L_{cd}$ model) 
over a \emph{noisy} beeping network that does not provide any inherent collision detection,~$BL_\eps$. 
Similarly, we can simulate any of the weaker beeping models (i.e., $BL$, $B_{cd}L$, $BL_{cd}$)
over the noisy beeping model~$BL_\eps$. 
Our simulation has a polynomially small error probability in~$nR$
and has a multiplicative overhead of~$O(\log n+\log R)$ in the number of rounds,
where $R$ is the number of rounds in the protocol to  be simulated.

\begin{theorem}
\label{thm:BeepSim}
Any protocol $\pi$ that takes $R=|\pi|$ rounds in the $B_{cd}L_{cd}$ model
can be simulated in the $BL_\eps$ in  $R \cdot O(\log n+\log R)$ rounds
and a success probability of at least $1-2^{-\Omega(\log n+\log R)}$.
\end{theorem}
Note that the above implies that a simulation with similar parameters is possible for $\pi$ defined in either the $B_{cd}L$, $BL_{cd}$, or the $BL$ models, by simply ignoring the information the parties learn about collisions that have happened.
\begin{proof}
This claim follows by simulating each round of the original protocol $\pi$ by performing the \textsf{CollisionDetection} procedure (Algorithm~\ref{alg:collision}).
For a given simulated round, any party that wants to beep in~$\pi$ is set $active$ in \textsf{CollisionDetection}, or otherwise it is passive. 
The output of each round is whether $0$, $1$, or $>1$ parties were active in this round, which corresponds to the round outcome of the node in the $B_{cd}L_{cd}$ model.  

Every round of~$\pi$ is simulated successfully with probability $1-2^{-\Omega(n_c)}$ and the entire $R$-round computation succeeds  with probability bounded by~$1-R\cdot2^{-\Omega(n_c)}$.
By setting $n_c=\Theta(\log R + \log n)$ we guarantee that the simulation succeeds up to a polynomially small error probability in~$n$ and~$R$. 
\end{proof}
Note that the above method requires the parties to know in advance the length of the protocol~$R$ (or a reasonable bound on it).

The above implies that our simulation is essentially tight. When $|\pi|= n^{O(1)}$, our overhead is $\Theta(\log n)$, which is tight for certain tasks, such as coloring (see Section~\ref{sec:mis} below), or the collision-detection task which is trivial in the $B_{cd}L_{cd}$ model, and takes~$\Theta(\log n)$ in $BL_\eps$ (Corollary~\ref{cor:collsion}). %
When $|\pi| >  n^{\omega(1)}$ we obtain an overhead of $O(\log n+ \log |\pi|)$. The question of whether such an overhead is tight remains open.

\subsection{Applications and Implications}\label{sec:applicaitons}
The simulation procedure in Section~\ref{sec:simulaiton} 
allows us to obtain noise-resilient protocols in the beeping model for various tasks
such as coloring, MIS, and leader election. These are obtained by simulating 
a (standard) beeping-model protocol for the required task. 
One main advantage is that we can simulate a protocol that was designed for the
$B_{cd}L_{cd}$ model. 
These usually feature reduced complexity due to the stronger model 
assumption of (inherent) collision detection. 
In many situations, simulating the $B_{cd}L_{cd}$ protocol rather than a
$BL$ protocol yields a noise-resilient protocol with reduced complexity.

To illustrate, consider the task of coloring (see definition below).
It is known that $O(\Delta+\log n)$ rounds suffice to color a network whose maximum degree is~$\Delta$
assuming (beeping) collision detection ($B_{cd}L$), while $\Omega(\Delta \log n)$ rounds are required if no collision detection is assumed.
Our noise-resilient simulation can make use of the more efficient $O(\Delta)$-round  $B_{cd}L$-protocol as its noiseless underlying algorithm and shave a $\log n$ factor off the overall round complexity. 
This choice allows us to achieve tight bounds for the task of coloring over the noisy beeping model. 
Unfortunately, the same approach does not lead to tight bounds for other tasks.

In the next subsection we consider the tasks of coloring, MIS, and leader election, 
and discuss noise-resilient protocols for these tasks. We also provide information about best known lower bounds for these tasks in the noiseless model, since it carries over to the noisy setting.
In Table~\ref{tbl:summary} we summarize the protocols implied by our simulation process for the noisy-beeping model.

\begin{table}[ht]
\vspace{1ex}
\begin{center}
\small
\begin{tabular}{llll}
\textbf{Task} & \textbf{$BL_\eps$ Upper Bound} & \multicolumn{2}{l}{~~~~\textbf{$BL_\eps$ Lower Bound}}  \\
\hline\hline\\[-2ex]
Collision Detection ~~~ & $O(\log n)$ & $\Omega(\log n)$ &  \parbox[c]{2cm}{\small this paper, \\\cite{AABCHK13}, \\  \cite{CMRZ19} } \\ \\[-2ex]
\hline 
Coloring   {\rule{0pt}{2.4ex}}   %
	& $O(\Delta \log n+\log^2n)$  & $\Omega(\Delta \log n)$    &{\small \cite{CMT17}}  \\ %
MIS %
 & $O(\log^2 n)$ & $\Omega(  \log n)$ & {\small\cite{MRSZ11}}\\
Leader Election &  $O(D\log n + \log^2n)$ ~~~~ & $\Omega(D+\log n)$   & {\small \cite{GH13}}
\end{tabular}
\end{center}
\caption{\small Summary of the results for the noisy-beeping model~$BL_{\eps}$. Upper bounds hold with high probability. In the communication network, $D$ is the diameter; $\Delta$ is its maximum degree; $n$ is the number of nodes.}
\label{tbl:summary}
\end{table}

\subsubsection{Coloring}
\label{sec:coloring}
The task of coloring a graph consists of assigning each node $v\in V$ with a color $c(v) \in K$ such that no two neighboring nodes are assigned the same color, i.e., $\forall (u,v)\in E, c(v) \ne c(u)$. 
The set $K$ represents the possible colors and its size is larger than the maximum degree of the graph; it is usually assumed that $K=O(\Delta)$. 
Certain algorithms allow a larger number of colors, e.g., $K=O(\Delta + \log n)$, 
an assumption that potentially  makes the  round complexity smaller. 
Protocols are allowed to use randomness and are required to succeed with high probability.

The current state of the art in coloring protocols for the beeping model is as follows.
For the $B_{cd}L$ model, Casteigts et~al.~\cite{CMRZ19} provided a randomized algorithm 
with $\Theta(\Delta+\log n)$ rounds and $K=O(\Delta+\log n)$ colors. If the number of colors $K\ge \Delta$ is known to all parties, an algorithm with round complexity~$O(K\log n)$ is known.
In the $BL$ model, Cornejo and Kuhn~\cite{CK10} designed a protocol with $O(\Delta \log n)$ rounds assuming knowledge of $K=O(\Delta)$. Note the $\log n$ gap between the $BL$ model and $B_{cd}L$ model.

Simulating the protocol of~\cite{CMRZ19} using Theorem~\ref{thm:BeepSim} yields the following.
\begin{theorem}
There exists an efficient randomized coloring protocol in the noisy beeping model in $O(\Delta\log n + \log^2n)$ rounds that succeeds with high probability. 
The protocol employs $K=O(\Delta+\log n)$ colors and assumes that the parties only know a bound on the size of the network~$n$.
\end{theorem}

A lower bound for coloring in the $BL$ model was given by Chlebus et~al.~\cite{CMT17} who argued that any randomized coloring algorithm over $G=K_n$ requires~$\Omega(n\log n)$ rounds. 
This lower bound carries over to the noisy setting, and proves that our noise-resilient coloring simulation is optimal.

\subsubsection{MIS}
\label{sec:mis}
A \emph{Maximal Independent Set} (MIS) is a set of nodes $I \subseteq V$ such that (a) no two nodes in $I$ are neighbors, and (b) any node $v\in V$ is either in~$I$ or a neighbor of some node in~$I$.
Finding an MIS is one of the fundamental tasks in distributed computing as this set ``governs'' the entire network: every node is an immediate neighbor of some member of the MIS, while there is only little redundancy in the sense that no two ``governor'' nodes are connected to each other.

Algorithms for finding MIS in the beeping model started with the work of Afek et al.~\cite{AABHBB11} which achieved $O(\log^2 n)$ rounds in~$B_{cd}L$. The current state of the art is an algorithm with round complexity of $O(\log n)$ in the $B_{cd}L$ model by Jeavons et al.~\cite{JSX16} (constants later improved by~\cite{CMRZ19}).

Simulating the protocol of~\cite{CMRZ19} using Theorem~\ref{thm:BeepSim} yields the following.
\begin{theorem}
There exists an efficient randomized protocol solving the MIS problem in the noisy beeping model in $O(\log^2n)$ rounds and with high probability of success. 
\end{theorem}

As for a lower bound for the MIS task, M{\' e}tivier et al.~\cite{MRSZ11} observed that the work of Kothapalli et~al.\@ on distributed coloring~\cite{KSOS06} implies an $\Omega(\log n)$ lower bound on the complexity of finding MIS in the beeping model, through a reduction by Wattenhofer~\cite{Watt07}.

\subsubsection{Leader Election}
\label{sec:le}
Many distributed tasks begin at some designated node, \emph{a leader}, 
whose role is to initiate or coordinate the progress of the required distributed task. 
The task of electing a leader 
assumes that the nodes begin with some identifier (possibly chosen at random by the node itself), otherwise deterministic leader election is sometimes impossible~\cite{AttiyaW2004}. 
At the end of the protocol all nodes must output the same identifier of a node, elected to be the leader. 
This task basically amounts to breaking symmetry among all participants.

The state of the art is a leader election algorithm in the $BL$ beeping model by 
Dufoulon, Burman and Beauquier~\cite{DBB18}
which has a round complexity of $O(D + \log n)$ where $D$ is the diameter of the network. Simulating this protocol gives:
\begin{theorem}
There exists an efficient randomized leader election protocol in the noisy beeping model that takes $O(D\log n+\log^2n)$ rounds and succeeds with high probability. 
\end{theorem}

A lower bound 
of $\Omega(D+\log n)$ rounds for electing a leader 
follows from~\cite{NO02,GLS12}, see also~\cite{GH13}.

\goodbreak 
\section[Simulation of message passing protocols]{Simulation of message passing protocols over $BL_\eps$ networks}
\label{sec:mp}

In this section we show how to simulate a protocol designed for message-passing networks (in the CONGEST model) over noisy beeping networks. Our simulation succeeds with high probability and incurs a multiplicative communication overhead that tends to $O(\Delta\cdot \min(\Delta^2,n) )$ as the message-passing  protocol's length grows to infinity. 
This implies a \emph{constant communication overhead} for constant-degree networks.

Furthermore, we show that this simulation is essentially tight. That is, we show that over a clique, some tasks necessitate a multiplicative overhead of $\Omega(n^2)$, which matches the simulation's upper bound.

\paragraph{The message-passing CONGEST.}
The CONGEST($B$) model assumes a network abstracted as the undirected graph~$G=(V,E)$ with $n=|V|$ nodes, where each $(u,v)\in E$ is a bi-directional communication channel. Communication works in synchronized rounds where in every round, a message of $B$~bits is potentially communicated over every channel. Commonly, $B=O(\log n)$.

A \emph{protocol} in this model determines for each node~$v$, which message it needs to send to each of its neighbors~$u$ in every round. The message is denoted $M_{vu}$ 
and its length satisfies~$|M_{vu}|\le B$. Messages to different neighbors (in any given round) may be different. 
However, we assume that in every round exactly $2|E|$ messages are being communicated in the network; such protocols are sometimes called \emph{fully utilized}.
Recall that $|\pi|$ denotes the number of rounds the protocol~$\pi$ takes, that is, all parties terminate after exactly $|\pi|$~rounds, and this quantity is assumed to be known to the parties in advance.

Similar to the beeping model,  we also assume that nodes do not  have unique identifiers  in the CONGEST model. Instead, each node has a list of \emph{ports} to which it can send messages, where each port is connected to a single neighbor. In particular, port numbers may be arbitrary and no binding between port numbers and nodes identities may be assumed.

\subsection{The Simulation}
We now show how to simulate a CONGEST(B) protocol over the noisy beeping $BL_\eps$ model. 

The simulation idea is as follows (cf.~\cite{BBDK18,ABLP89}).
We first employ a 2-hop coloring over the network. Next, each node associates the identities of its neighbors (i.e., their port number) with their color. 
The 2-hop coloring allows us to avoid collisions through time-division, 
by allowing exactly a single ``color'' to beep at any given time. 
Every round of~$\pi$ can be simulated by $O(c\cdot \Delta \cdot B)$ rounds in~$BL_\eps$, 
where $c$ is the maximal number of colors.
To that end, the node $v$ concatenates the messages $\{M_{vu}\}_{u\in \neigh_v}$ that it needs to send in that round to a single string $M_v$ and broadcasts this message, encoded with some error correction code. After $O(B\Delta)$ rounds a single node $v$ will have completed communicating its~$M_v$, and thus after $O(cB\Delta)$ all the parties will have broadcast their messages without interfering with each other. 
Finally, we employ the multiparty interactive coding by 
Rajagopalan and Schulman~\cite{RS94} (see also~\cite{ABEGH19}) 
that will take care of noise corrupting received messages. This interactive coding is given by the following theorem.
\begin{theorem}[\cite{RS94,rajagopalan94}]\label{thm:RS}
For any integer $t \ge 0$, any fully-utilized deterministic $R$-round protocol~$\pi$ in the message-passing model over any network~$G$ communicating bits
can efficiently be transformed into a fully-utilized noise-resilient protocol~$\Pi$ in the message-passing model that simulates~$\pi$. 
The protocol~$\Pi$ communicates messages of constant size,
takes $2R+t$~rounds and succeeds with probability
$1-n(2(\Delta+1)p)^{\Omega(R+t)}$ given that any message is correctly received with probability~$1-p$. 
\end{theorem}

The simulation of CONGEST(B) over $BL_\eps$ is formally described in Algorithm~\ref{alg:simcongest}. 
This scheme already assumes that the nodes possess a 2-hop coloring with $c$~colors. 
Also note that the parties in Algorithm~\ref{alg:simcongest} are assumed to know the maximum degree~$\Delta$ of~$G$. 
This information (or a reasonable bound) can be derived from~$c$.

We note that obtaining a 2-hop coloring in $BL_\eps$ can be done in $O(\Delta^2\log n + \log^2n)$ rounds with $c=O(\Delta^2+\log n)$ colors, or in 
$O(\Delta^2\log^2 n)$ rounds with $c=\Delta^2+1$ colors, assuming $\Delta$ is known, 
both via a $B_{cd}L_{cd}$ scheme by~\cite{CMRZ19} and Theorem~\ref{thm:BeepSim}.

\newcommand{\LongLine}{\parindent=-1em\leftskip=1em}

\begin{algorithm}[htp]
\caption{Simulation of CONGEST(B) over $BL_\eps$ for node~$u$}
\label{alg:simcongest}
\begin{algorithmic}[1]
\small

\State \textbf{Input:} 
a 2-hop coloring with maximal color $c$; A protocol $\pi$ in CONGEST(B); 
$\Delta$ the  maximum degree of the communication network~$G$.
\Statex

\State Let $C:\{0,1\}^{k_C} \to \{0,1\}^{n_C}$ be a code with $k_C=\Theta(\Delta)$,  $n_C=\Theta(\Delta)$, and a constant relative distance.

\State Let $\pi_1$ be a simulation of $\pi$ in CONGEST(1). Namely, split each transmission of $\pi$ into (exactly) $B$ transmissions of~$\pi_1$, where nodes in $\pi_1$ communicate \emph{bits}.
\State Let $\Pi$ be the encoded protocol of~$\pi_1$ given by Theorem~\ref{thm:RS} with message alphabet $\Sigma$ where $|\Sigma|=O(1)$, and $t = \max(0, \Theta(\tfrac1\Delta\log n) -B|\pi|)$. 
\State Set $k_C \ge \Delta \log |\Sigma|$.

\Statex

\Statex \textbf{Preprocessing Step: }

\State \hspace{\algorithmicindent} 
\parbox[t]{0.9\columnwidth}{\LongLine Node $u$ ``collects'' its \emph{colorset}---the colors of all its neighbors $\neigh_u$. Without noise this can be done in $c$ rounds, where each node beeps in the slot assigned with its color. 
Via Theorem~\ref{thm:BeepSim}, this can be done in a noise-resilient manner in $O(c \log n)$ rounds.}
\State \hspace{\algorithmicindent} 
\parbox[t]{0.9\columnwidth}{\LongLine{}Every node $u$ learns the colorset of every~$v\in \neigh_u$. Without noise this can be done in $c^2$ rounds: we give each color $c$ slots; each node beeps 1 for every color in $[c]$ that appears in its colorset. Via Theorem~\ref{thm:BeepSim}, this can be done in a noise-resilient manner in $O(c^2\log n)$ rounds.\strut}
\State \hspace{\algorithmicindent} 
\parbox[t]{0.9\columnwidth}{\LongLine Let $C_u$ be the colorset of $u$. Fix an (arbitrary) mapping between the color of $u$'s neighbors and their port number.}
	
\Statex \Statex
	
\Repeat
	\For {$i \in [c]$ in a round-robin manner}
		\If {self color is $i$}
			\State \parbox[t]{0.6\columnwidth}{\LongLine{}Let $\overline M= M_{u,j_1}\circ \cdots \circ M_{u,j_t}$ for $\{j_1,...,j_t\} = C_u$, be the concatenation  of the $t\le \Delta$ messages $M_{uj}$ that $u$ needs to send to its neighbors in~$\Pi$ %
			(pad with zeros so that $|\overline M|=\Delta|\Sigma|$).
			Order the $M_{uj}$ in $\overline M$ by an increasing color number, $j_1< j_2 < \dotsb < j_t$.\strut} \Comment{at most $\Delta$ messages\strut}
			\State  Beep $C(\overline M)$    \Comment{$n_C$ rounds}
		\Else 
			\State Listen for $n_C$ rounds, decode the received codeword via $C^{-1}$.
			\State \parbox[t]{0.6\columnwidth}{\LongLine From the decoded message $\overline M$ obtain $M_{ij}$ for $j$ the color of~$u$  \\ (if such a message exists).\strut}
			\Comment{\smash{\parbox[t]{0.2\columnwidth}{ $u$ knows the colorset of the sender, thus knows to split $\overline M$ into $M_{ij}$s.}}}
			\State Use $M_{ij}$ as the received message from color (port)~$i$ in~$\Pi$. 
		\EndIf
	\EndFor

\Until {$\Pi$ terminates}

\end{algorithmic}
\end{algorithm}

\subsection{Analysis}
In this section we prove that our simulation succeeds with high probability, and has the desired overhead.

\begin{theorem}\label{thm:MPsim}
Given a 2-hop coloring with $c$ colors and for any $\eps<1/2$, 
any fully-utilized protocol~$\pi$ in the CONGEST(B) model can be simulated 
with high probability in the noisy beeping model $BL_\eps$ 
in $O(c^2\log n) + \max(|\pi|,\log n/\Delta)\cdot O(B\cdot c \cdot \Delta)$~rounds. 
\end{theorem}

\begin{proof}
We argue that Algorithm~\ref{alg:simcongest} satisfies the theorem's statement. 

Let us begin by analyzing the round complexity of Algorithm~\ref{alg:simcongest} on an input~$\pi$.
The preprocessing step takes $O(c^2\log n)$ rounds. Note that this part occurs only once and does not affect the multiplicative blowup of the simulation.
The bit-version protocol, $\pi_1$, takes $|\pi_1|=B|\pi|$ rounds and communicates bits.
Via Theorem~\ref{thm:RS}, we know that $\Pi$ takes $O(|\pi_1|)=O(B|\pi|)$ rounds and uses symbols of constant size.  
The repeat loop of Algorithm~\ref{alg:simcongest} occurs $|\Pi|$ times, where each iteration consists of $c$ transmissions of codewords from~$C$. Transmitting a single codeword from~$C$ takes $n_c=O(\Delta)$ rounds of the simulation. Putting it all together, the statement holds. 

We now move on to the correctness analysis. 
First, note that in this analysis, we identify colors with port numbers, 
i.e., a message sent to port $i$ in $\pi$ is sent to the (neighbor) node with color $i$ in the simulation. 
Note that the coloring need not be contiguous or surjective, but we assume all nodes know their and their neighbors' colorsets (recall that a \emph{colorset} of a node $u$ is set of the colors of its neighbors $\neigh_u$). The preprocessing step satisfies this assumption. 
The mapping between ports and colors may be arbitrary, but this has no effect on the correctness of~$\pi$, since  port numbers are arbitrary to begin with.

\begin{lemma}\label{lem:singleRound}
Every iteration of the simulation's repeat-loop simulates a single round of $\Pi$ with message-error probability $2^{-\Omega(\Delta)}$.
\end{lemma}
\begin{proof}
The coding of Theorem~\ref{thm:RS} produces a fully-utilized protocol in which in each round, every party sends one message to each of its neighbors. Hence, a simulation of a single round of $\Pi$ requires delivering a single message from $u$ to $v$ for any pair of neighbors $(u,v)\in E$. 
Note that two messages are delivered on each edge: one in each direction.

The simulation of a single round takes $c$ epochs, where in each epoch a node with a matching color
sends a single message~$\overline M$ that encodes its messages to all its neighbors.
Furthermore, the 2-hop coloring guarantees that no two nodes in the neighborhood of any node~$u$ have the same color.
Hence, every node $u$ hears only a single sender at every given epoch, and no collisions occur.
Thus, the node~$u$ gets the message~$\overline M$ from $v\in \neigh_u$ with color~$i$ during the $i$-th epoch of the cycle, 
while no other nodes in $u$'s neighborhood beep during that same epoch. 
Node~$u$ can parse~$\overline M$ and obtain the specific message~$M_{vu}$ that $v$ addressed to~$u$ since $u$ knows its own color, and it knows the colorset of~$v$. With this information $u$ can infer the position of $M_{vu}$ in~$\overline M$.

Nonetheless, noise may alter the communicated symbols. Encoding each message~$\overline M$ with the error-correcting code~$C$ that has a constant relative distance guarantees that each such message is decoded correctly, with a probability of at least~$1-2^{-\Theta(n_C)} = 1-2^{-\Theta(\Delta)}$. 
\end{proof}

By a repeated activation of Lemma~\ref{lem:singleRound} we get that all the $|\Pi|$ rounds of $\Pi$ are simulated in $BL_\eps$ without any collisions, and that the message-error probability is at most $2^{-\Theta(\Delta)}$. Theorem~\ref{thm:RS} then states that $\pi_1$ (and hence, $\pi$) is correctly simulated,  with error probability at most
\[
n\left (2(\Delta+1) 2^{-\Theta(\Delta)}\right)^{\Omega(|\pi_1|+t)} = 2^{-\Omega(B\Delta \cdot |\pi|+\Delta t)+\log n}.
\]
To complete the proof we need that, for some constant $\alpha$ hidden in the $\Omega()$ notation,  $\alpha(B\Delta |\pi|+\Delta t) > \log n$. This means that
$t > \frac{\log n}{\alpha\Delta} - B|\pi|$ suffices to guarantee success with high probability.
Note that when $|\pi|\to\infty$ in particular $|\pi|>O(\log n/\Delta)$, we have $t=0$ and the overall simulation takes $O(c^2\log n) + |\pi| \cdot O(B\cdot c \cdot \Delta)$, i.e., the multiplicative overhead is~$O(Bc\Delta)$.
Otherwise (when $|\pi| < O( \log n/\Delta)$), the factor $t$ increases the length of~$\Pi$  to be of magnitude $\Theta(\lceil \log n / \Delta\rceil)$ which guarantees the success with high probability when $n\to\infty$. The simulation in this case will take $O(c^2\log n) + \Theta(\lceil \log n / \Delta\rceil)\cdot O(B c \Delta)$ rounds.
\end{proof}  %

\begin{remark}
Despite the fact that the coding of Theorem~\ref{thm:RS} is computationally inefficient,
note that the simulation of Algorithm~\ref{alg:simcongest} can be made efficient. 
To this end, one needs to employ an efficient (randomized) version of the Rajagopalan-Schulman coding
such as the one described in~\cite{GMS14} or as described in~\cite{ABEGH19}.
\end{remark}

\subsection[Optimality for K\_n]{Optimality for $K_n$}

In this section we argue that our simulation of message-passing protocols 
is optimal for certain graphs. 
In particular, we show that if each party in $K_n$ starts with $n-1$  random (bit-)messages, where the $i$-th message is designated to 
its $i$-th neighbor, then $\Theta(n^2)$ overhead is necessary and sufficient.
That is,  a CONGEST(1) protocol for the above task will take a single round while the $BL$ protocol necessitates $\Omega(n^2)$ rounds, yielding the desired simulation overhead. 

The high-level intuition is that over the network~$K_n$, each party hears the beeps of all other parties. We can think of this graph as being a single beeping channel, which can only communicate (broadcast) a single bit of information in every round. 
To see this, consider a beeping protocol over~$K_n$ in the $BL$ model. Consider the input to the channel of the $i$-th party, $X_i\in \{0,1\}$, and the superimposed signal transmitted by the channel to all parties $Y = X_1 \vee \cdots \vee X_n$. 
There are two options: if  $X_1\cdots X_n=0^n$ then $Y=0$, and if  $X_1\cdots X_n \ne 0^n$ then $Y=1$. 
This implies that the channel can transfer only a single bit at a time---whether or not \emph{all} the parties had the input~$0$.

Let us formally define the $k$-message-exchange task, which is a sequential $k$-fold generalization of the exchange task discussed above.
\begin{definition}
The $k$-message-exchange task is defined as follows. 
\textbf{Input:}  party $i$ is given $k$ sets of $n-1$ messages, $M^1_i$ to $M^k_i$, where
 $M^1_i = M^1_{i,1},\ldots,M^1_{i,n}$, etc.\footnote{Since each node has $n-1$ neighbors to send messages to, the value of $M^t_{i,i}$ is undefined by the input and need not be communicated. We can assume $M^t_{i,i}=0$ for any $t\in[k],i\in[n]$, which is added just to ease the notations.} Assume that each  $M^t_{i,j}\in\{0,1\}$ is uniformly distributed, independently of all other messages.
\textbf{Output:} 
party $i$ needs to output the $k$ vectors $\tilde M^1$ to $\tilde M^k$, where
$\tilde M^1_i=M^1_{1,i},M^1_{2,i},\ldots,M^1_{n,i}$, etc.
\end{definition}
Note that the $k$-message-exchange task can be solved over~$K_n$ by a trivial protocol~$\tilde\pi$ in CONGEST(1),
where party~$i$ simply sends $M^t_{i,j}$ to party~$j$ in round~$t$. Hence, for any $n\in\mathbb{N}$,  $|\tilde\pi(K_n)|=k$ rounds. 

The reason we define the $k$-message-exchange version of this task rather than considering a single round is because we want our result to apply for the case where  $n\to\infty$ as well as for $|\pi|\to\infty$ (that is, $k\to\infty$). This notion is somewhat stronger than the notion of $n\to\infty$ commonly used in distributed computation. This stronger notion applies for example to a fixed network (fixed~$n$) over which we want to perform some protocol. If the protocol is short, standard error correction (or repetition) of each message will guarantee the correct simulation of the protocol. However, when the protocol is long (e.g., when $k\to\infty$ rounds are to be simulated), the per-message coding will eventually fail, and stronger coding techniques must be employed (e.g., the techniques of Theorem~\ref{thm:RS}).

\begin{theorem}
Solving the $k$-message-exchange task over~$K_n$ in the $BL$ model requires~$\Theta(k n^2)$ rounds, when either $k\to\infty$ or when $n\to\infty$.
\end{theorem}
\begin{proof}
We begin with the lower bound. Note that since all parties are connected to the same beeping channel (since we are on a clique), all parties hear the same transmitted transcript. That is, all information is broadcast to all parties. We show the lower bound by reducing the problem of \emph{multisource broadcast}~\cite{BII93,LBA93TR} to the problem of $k$-message-exchange; then, a lower bound by~\cite{CD19bc} to the multisource broadcast problem will imply our desired bound.

Let us first define the multisource broadcast problem and then define the reduction. In the $k'$-multisource broadcast problem we are given a network $G'$ with $n'$ parties of diameter~$D'$. The nodes are given unique ID's from the range~$[L']$.  There are $k'$ nodes which begin with a message from the range~$[M']$. 
The goal is to make all the parties learn all the $k'$ messages. Previous work~\cite{CD19bc} distinguished two flavors of this task: in the multisource broadcast \emph{with provenance} each node needs to learn the $k'$ messages along with the ID of the source node of each such message, i.e., learn the $k'$ pairs $(ID_i, message_i)$ for $i\in[k']$. In the multisource broadcast \emph{without provenance} one only needs to learn the messages themselves, without having to know their source ID or their multiplicity, i.e., learn the set $\{ message_i \mid i\in[k']\}$. We focus on the `with provenance' version.

We use the following lower bound by Czumaj and Davies~\cite{CD19bc} on the number of rounds any $k'$-multisource broadcast over network $G'$ in the $BL$ model must take to complete with high probability.
\begin{lemma}[{\cite[Theorem~18]{CD19bc}}]\label{lem:MB_failureProb}
In all ranges of parameters, any algorithm achieving multi-broadcast with provenance
within $\frac{k'}{8} \log\frac{L'M'}{k'}$ rounds
fails with probability at least $1-\big(\frac{L'M'}{k'}\big)^{-k'/134}$.
\end{lemma}
From the above it is immediatly clear that $\Omega(k'\log\frac{L'M'}{k'})$ rounds are required to succeed with probability $1-o(1)$, both when $n\to\infty$ (and thus $L'\to\infty$) or when $M'\to\infty$, assuming the other parameters are fixed.

We now argue that the $k'$-multisource broadcast with provenance over the clique $G'=K_{n'}$, with a specific range of parameters $n',k',L',M'$ that we describe below, can be reduced to the $k$-message exchange problem over the clique $G=K_n$ (with parameters $n,k$ to be defined).
Fix $n,k$ arbitrarily. 
Set $n'=n$ and consider the clique $G'=K_{n'}$, where each one of the parties owns a unique ID out of the range~$[L']=[n]$,
and where every node is a source ($k'=n'$) that holds a message $m_i\in[M']$ to be broadcast, with $\log M'=k\cdot (n-1)$. 
For any $i\in[n]$ construct the messages $M^1_i,\ldots, M^k_i$ so that the binary representation of $M_i$ equals the binary representations of~$m_i$. Perform $k$-message-exchange with the above $M_i$s. Note that over $K_n$ in the $BL$ model, all the parties see the same transcript---the same history of rounds with beeps and rounds with silence.
Therefore, if some party $i$ learns some message $\tilde M_i$, then all the other parties learn the same information as well. This means that all parties learn all~$M_i$s. Furthermore, note that learning $\tilde M_i$ means learning the provenance of each bit, e.g., that $\tilde M_{i,j}$ originated at source~$j$. Thus, each party~$i$ can restore the provenance of each bit it learns in~$\tilde M_i$ but also for any other $\tilde M_{j}$ with~$j\ne i$.

Now, assume that the above $k$-message exchange takes $o(k n^2)$ rounds, then the $k'$-multisource broadcast with provenance over the set of parameters $k'=n,L'=n,\log M'=k(n-1)$ completes in $o(k  n^2)$ rounds, yet $k' \log (L'M'/k') = n \cdot(k(n-1))=\Theta(k n^2)$ which contradicts Lemma~\ref{lem:MB_failureProb} for this set of parameters, when either $n\to\infty$ or $k\to\infty$.

\medskip
We complete the proof of the theorem by showing a matching upper bound of $O(kn^2)$ for the same task over~$K_n$.
Consider Theorem~\ref{thm:MPsim} and recall that a 2-hop coloring can be done with $c=n$ colors, i.e., every node gets a different color. This can be obtained using the method by Chlebus et~al.~\cite{CMT17} with high probability in $O(n\log n)$ rounds in the noiseless $BL$ model and in $O(n\log^2n)$ rounds in~$BL_{\eps}$, via Theorem~\ref{thm:main}. Furthermore, note that since we are over a clique, all the parties learn the coloring and the pre-processing steps of collecting the colorset are no longer needed. 
Thus, in this special case, pre-processing takes only $O(n\log^2 n)$ rounds  instead of $O(c^2\log n)$. 

Given the above (2-hop) $n$-coloring and the adjusted pre-processing step, we get that any CONGEST(B) protocol~$\pi$ can be simulated with high probability 
in $O(n\log^2 n)+ |\pi|\cdot O(B\cdot c \cdot \Delta)$ rounds of~$BL_\eps$.
Since the $k$-message exchange task over $K_n$ can be solved in $k$ rounds of CONGEST(1), it can be simulated with high probability over $BL_\eps$ 
in $O(n\log^2 n)+k\cdot(1\cdot n \cdot (n-1))=\Theta(kn^2)$ rounds, which gives the desired bound.
\end{proof}
Note that the above lower bound is stated in the $BL$ model and the upper bound in the $BL_\eps$ model. Thus, the statement holds for both the $BL$ and the $BL_\eps$ models.

\bibliographystyle{alphabbrv-doi.bst}	%
\bibliography{network}

\end{document}